\documentclass[journal]{IEEEtran}
\usepackage{graphicx}
\usepackage{caption}
\usepackage{subcaption}
\usepackage{amsthm}
\newtheorem{theorem}{Theorem}
\usepackage{amsmath}

\ifCLASSINFOpdf
\else
\fi
\begin{document}
%
\title{Carry Value Transformation (CVT) - Exclusive OR (XOR) Tree and Its Significant Properties}
%

%

\author{Jayanta~Kumar~Das,
        Pabitra~Pal~Choudhury
        and~Sudhakar~Sahoo
\thanks{J. K. Das and P. Pal Choudhury are in Applied Statistic Unit, Indian Statistical Institute, India,
203, B.T Road Kolkata-700108, E-mail:dasjayantakumar89@gmail.com}
\thanks{S. Sahoo with Institute of Mathematics and applications, India, Orissa, Bhubaneswar-751003}
}

\maketitle

\begin{abstract}
CVT and XOR are two binary operations together used to calculate the sum of two non-negative integers on using a recursive mechanism. In this present study the convergence behaviors of this recursive mechanism has been captured through a tree like structure named as CVT-XOR Tree. We have analyzed how to identify the parent nodes, leaf nodes and internal nodes in the CVT-XOR Tree. We also provide the parent information, depth information and the number of children of a node in different CVT-XOR Trees on defining three different matrices. Lastly, one observation is made towards very old Mathematical problem of Goldbach Conjecture. 

\end{abstract}

\begin{IEEEkeywords}
Carry Value Transformation (CVT), Exclusive OR (XOR), CVT-XOR Tree, Fractal, Goldbach Conjecture.
\end{IEEEkeywords}

%
\IEEEpeerreviewmaketitle

\section{Introduction}
%
%
%
%
\IEEEPARstart{C}{arry}  Value Transformation (CVT) is one of the first and foremost transformation in the field of manipulating the strings of bits and other transformations of similar nature such as Extreme Value Transformation (EVT) [2], 2- Variable Boolean Operation (2-VBO) [4], Integral Value Transformations (IVTs) [5] and so on came after that. All these transformations have lots of applications in the field like pattern formations [2, 4], solving Round Rabin Tournament problem [6], Collatz-like functions [5] and so forth. In [7] with the help of Cellular Automata, CVT has been used for efficient hardware design of some basic arithmetic operations. As CVT is one of the important transformations in this area of study so further properties of CVT should have been thoroughly developed for some of its other future scope.  

In the field of data structure for the organization of non-linear data we have seen many  Tree structures like Binary Tree, AVL Tree, B Tree, B+ Tree etc. and many of their applications. In this paper, we have designed a new Tree structure named as CVT-XOR Tree in the domain of CVT and XOR. For this Tree two fundamental logics of CVT and XOR from [3] are used: 
(1) In [3] it is proved that addition of any two non-negative integers expressed as binary numbers is same as addition of their CVT and their XOR values. This result is also shown to be true for any base of the number system. 
(2) It has also been proved that in a successive addition of CVT and XOR of any two non-negative integers, the maximum number of iterations required to get either CVT=0 or XOR=0 is equal to the length of the bigger integer expressed as a binary string.

Organization of this paper is as follows: In section 2 some preliminaries of CVT and XOR along with a recursive algorithm for addition of two numbers are discussed. Section 3 deals with the fundamentals of CVT-XOR Tree and two approaches for the construction of CVT-XOR Tree. First part in section 4 discusses the different properties of CVT-XOR Tree and in second part some other significant properties of CVT-XOR Tree from the P, D and F matrices are enumerated. Section 5 deals with the conclusion of this paper along with some future research direction.
\section{Review of Earlier works on CVT}
CVT and XOR are two Transformations defined on a pair of non-negative integers expressed in binary. Interested reader can refer [1, 2, 3] for the definition of CVT which we are omitting for shortening the paper size. Some of the properties which we will be using in this paper are enumerated as follows:
\begin{description}
  \item[a)] CVT is always an even number and length is (n+1) bits.
\item[b)] It has been proved in [] that addition of two non-negative integers X and Y is equal to the sum of their CVT and XOR values i.e. X+Y=CVT(X, Y) + XOR(X, Y) and also the recurrence scheme always converges to (0, X+Y) in at most (n+1) iterations where n is the maximum number of bits required to represent the bigger number.
\end{description}
\textbf{\textit{Algorithm 1:}}
Recursive Algorithm for Addition of Two Non-Negative Integers:

$ \begin{array}{lc}
 ADD(X,Y) \\
 \{  \\
 if(X==0) \\
 return$  $Y \\
 elseif(Y==0)  \\
 return$ $ $ X $ \\
 else \\
 return$  $ADD(CVT(X,Y),$ $XOR(X,Y))\\
 
 \}
 \end{array}
$
 
As at the time of execution every recursive algorithm forms an Activation record in a Tree form so we are motivated to construct the Tree for the above recurrence algorithm named as CVT-XOR Tree and analyzed some of its significant features.
\section{Fundamentals of CVT-XOR Tree}
If we start from a pair of numbers as root node of the form (0, N) and back tracking from it using the iterative process X+Y=CVT(X, Y)+XOR(X,Y) [] we come across many intermediate nodes of the form (CVT, XOR) pair then we will get an m-ary tree. For a given number N, each node in the tree is represented by a pair or co-ordinate of the form (X, Y) such that X+Y=N. First co-ordinate i.e. X is the resultant CVT of two numbers and second co-ordinate i.e. Y is the resultant XOR of two numbers whose sum is N. As (CVT (0, N), XOR (0, N)) = (0, N) so  and this forms a self-loop at the root node. Thus a CVT-XOR Tree is like an m-ary tree with a loop only at the root node along with several non-leaf and leaf nodes. Depending of N value we have two types of Tree. If N is even, then the tree is called Even CVT-XOR Tree and If N is odd, then the tree is called Odd CVT-XOR Tree.

\textbf{\textit{Illustration 1:}} Let an even integer N=40 then (0, 40) is the root node and other leaf and non-leaf nodes in the form (X, Y) (such that X+Y=40) are shown in Fig 1. As 20 is an even number, so pair (X, Y) is in the form of either (even, even) or (odd, odd).
\begin{figure}[h]
\begin{center}
\includegraphics[scale=.10]{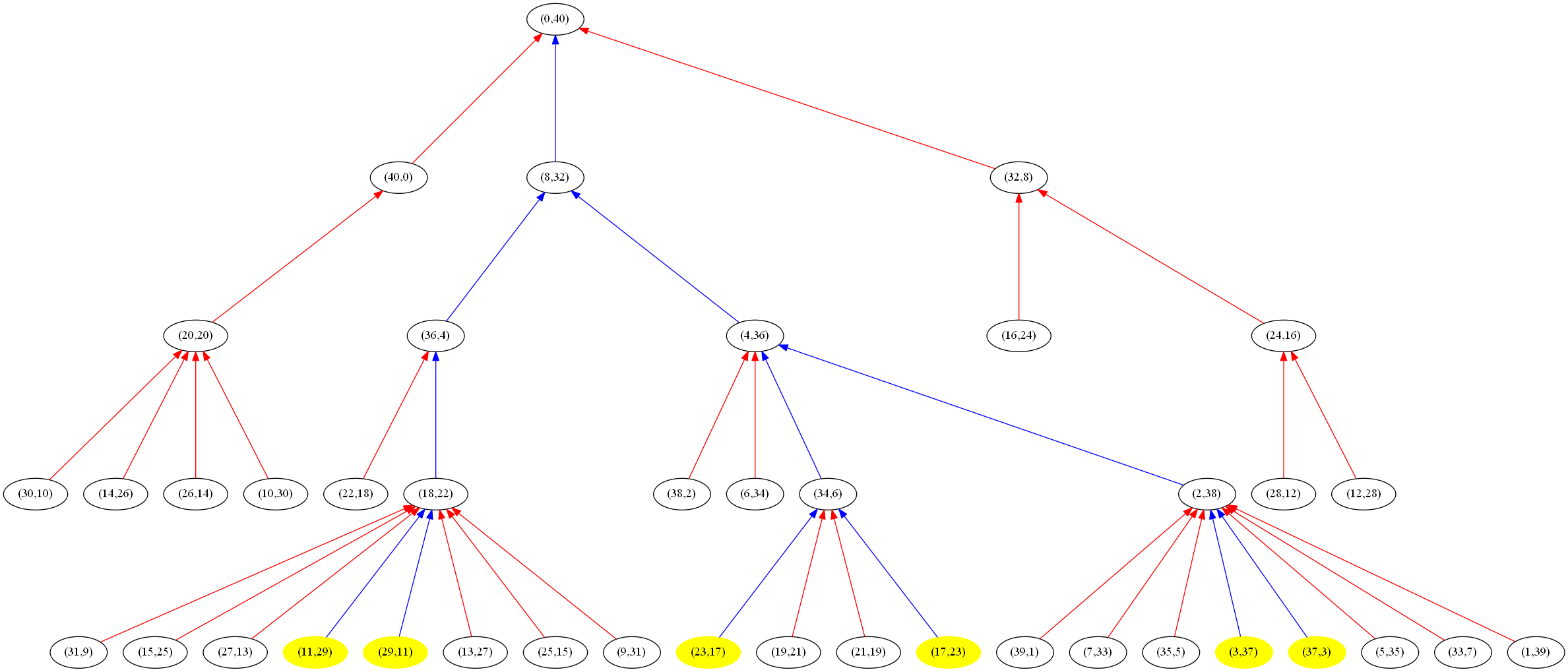}
\end{center}
\caption{CVT-XOR Tree for an even integer 40}
\label{fig:13}
\end{figure}

\textbf{\textit{Illustration 2:}} Let an even integer N=25 then (0, 25) is the root node and other leaf and non-leaf nodes in the form (X, Y) (such that X+Y=25) are shown in Fig 2. As 25 is an odd number, so pair (X, Y) is in the form of either (even, odd) or (odd, even).
\begin{figure}[h]
\begin{center}
\includegraphics[scale=.12]{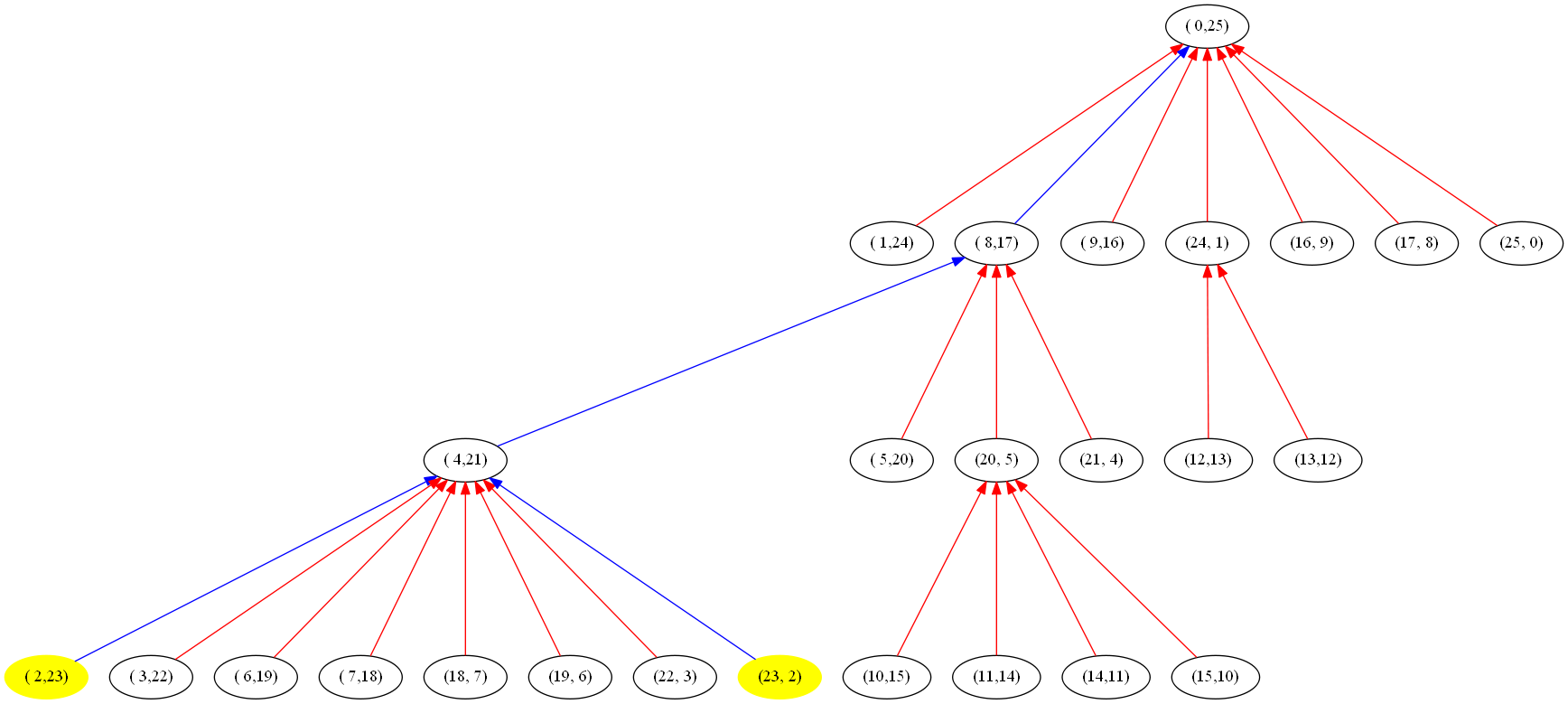}
\end{center}
\caption{CVT-XOR for an odd integer 25}
\label{fig:13}
\end{figure}
Using section 2.(b) the maximum height of the tree is  and is the number of significant bits required to represent N. Total number of nodes in CVT-XOR tree is N+1 which is discussed in section 4. Depending on the values of N either even or odd two types of trees can be analyzed named as Even CVT-XOR Tree or Odd CVT-XOR Tree.
As CVT of two numbers is always an even number so the first co-ordinate in the node is always even for all internal nodes. Therefore for these pairs, second co-ordinate should be an even number when N is even and is an odd number when N is odd. But when N is even the leaf nodes can be any pair of the form either (even, even) or (odd, odd) otherwise it is of the form (even, odd) or (odd, even). Following algorithm 2 shows whether a given node (X, Y) is leaf node or not.

\textbf{\textit{Algorithm 2:}}
Algorithm to Check whether a given node is leaf or not

$ \begin{array}{lc}
 NODE\_TYPE(X,Y) \\
 \{  \\
 if(X==0) \\
 $Print \textquotedblleft (X, Y) is root node\textquotedblright $ \\
 elseif(X==odd)  \\
 $Print \textquotedblleft (X, Y) is a leaf  node in X+Y CVTY-XOR Tree \textquotedblright $  \\
 else \\
 $Print \textquotedblleft (X, Y) may be a leaf node or an internal node$\\$further checking is required \textquotedblright  $\\
 
 \}
 \end{array}
$\\

Finding the structure of the tree for a given number N two different approaches are discussed below:
\subsection{Top down Approach}
The root is always of the form (0, N) this is consider in level-0 and using iterative process to find the nodes in level-1 we will search for the pair of integers X and Y such that X+Y=N, CVT (X, Y)=0 and XOR(X, Y)=N except node (0, N). All the nodes in this process are kept in level-1. This process is continuing until all the leaf nodes are found i.e. the nodes with no predecessor.

\subsection{Bottom up Approach} 
This process starts from all the leaf nodes whose sum is N. All the Leaf nodes for the number N can be found easily. For each such pairs we apply CVT and XOR operation and keep those nodes such that their pair sum is equal to N. Many nodes may converge towards the single node after CVT and XOR operations. This process is continued until we reach to the root node from all the leaf nodes which are in the form of (0, N).

In the CVT-XOR Tree if we backtrack from the root node to the leaf node we find different paths to reach to the leaf nodes at different depths. It seems that parent child relationship in the CVT-XOR Tree totally depends on characterization or bit patterns of the node pairs in the path. To understand these dynamics we have analyzed some of the properties of CVT-XOR Tree discussed in the following section.

\section{Significant Properties of CVT-XOR Trees}
\begin{theorem}
The Numbers of Nodes or Vertices in the CVT-XOR Tree is N+1.
\end{theorem}

\begin{proof}
Number of non-negative pair wise integer partitions of N is N+1 of the form (0, N), (1, N-1) … (N, 0). From section 2(b) all these pairs must converges to (0, N) so all these N+1 number of nodes are connected and will form a single CVT-XOR Tree. Hence the total numbers of Nodes in the CVT-XOR Tree is N+1.
\end{proof}
\textbf{\textit{Illustration 3:}}
Let an even integer N=18, we have 19 pairs as (0, 18), (18, 0), (1, 17), (2, 16), (3, 15), (4, 14), (5, 13), (6, 12), (7, 11), (8, 10), (9, 9), (10, 8), (11, 7), (12, 6), (13, 5), (14, 4), (15, 3), (16, 2), (17, 1) in CVT-XOR Tree shown in Fig 3.
\begin{figure}[h]
\begin{center}
\includegraphics[scale=.12]{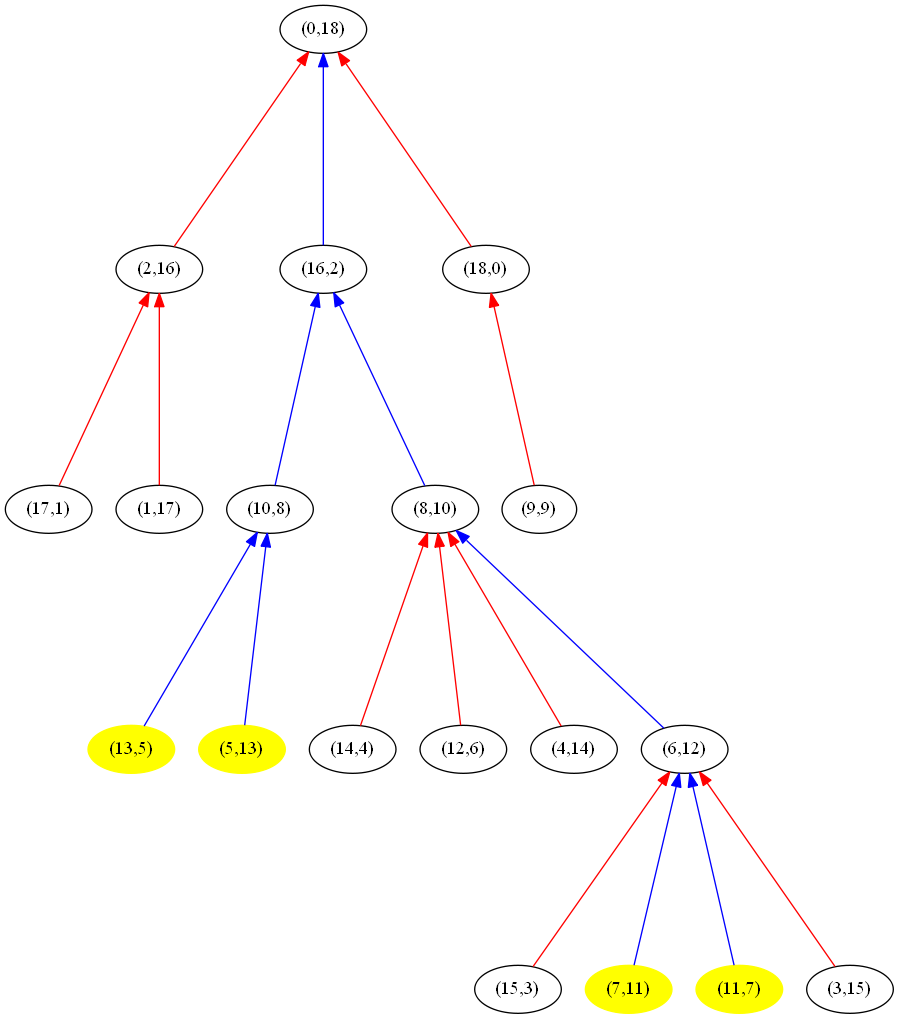}
\end{center}
\caption{CVT-XOR for an even integer 18}
\label{fig:18}
\end{figure}

\begin{theorem}
In CVT-XOR Tree if the pair (A, B) is in depth-d, then its symmetric pair (B, A) is also in depth-d except the root node.
\end{theorem}

\begin{proof}
Let,CVT(A, B)=P and XOR(A, B)=Q; where P and Q are the parent node (P, Q) for one child pair (A, B).\\
We know that, CVT(A, B)=CVT(B, A) and XOR(A, B)=XOR(B, A)\\
Then CVT(A, B)=CVT(B, A)=P and XOR(A, B)=XOR(B, A)=Q\\
That implies both the pairs (A, B) and (B, A) are the child nodes whose parent node is (P, Q) i.e. nodes (A, B) and (B, A) are two predecessor.\\
Therefore, (A, B)$\rightarrow$ (P, Q) and (B, A)$\rightarrow (P, Q)$. 
So they must belong to the same level i.e. in same depth. \\ 
According to the CVT properties, if the XOR value is zero in $i^{th}$ iteration, then the CVT value is zero in $i+1^{th}$ iteration. Therefore (N, 0) $\rightarrow$ (0, N) and (0, N) is the root node in CVT-XOR tree paradigm, So if (0, N) is in level-0 then the node (N, 0) will be always in level-1 in the CVT-XOR tree.
\end{proof}

\textbf{\textit{Illustration 4:}}
From Fig. 3, the pairs (2, 16) and (16, 2) are in lavel-1 similarly pairs (8, 10) and (10, 8) are in level-2 and so on. But the root node (0, 18) and its symmetric pair (18, 0) are in level-0 and level-1 respectively.

\begin{theorem}
For any pair (P, Q) where $P=a_{n+1} a_{n} a_{n-1}...a_{1}$ is an even number and $Q=b_{n+1} b_{n} b_{n-1}...b_{1}$ is either even or odd and if   $\exists$i; $1\leq i\leq n$ such that $a_{i+1}=b_{i}=1$, then the pair (P, Q) is the leaf node or contradictory (even, even) pair.
\end{theorem}

\begin{proof}
The theorem statement demands two binary variable $x_{i}$  and $y_{i}$  such that $AND(x_{i}, y_{i})=a_{i+1}=1$ and $XOR(x_{i}, y_{i})=b_{i}=1$ for some i, which is impossible. So the such pair (P, Q) must be a (even, even) leaf node.
\end{proof}
\textbf{\textit{Illustration 5:}}
As per the definition of CVT, we know that CVT of any two numbers is always an even number. Therefor there is no predecessor for an odd pair. But some (even, even) contradictory pairs are also in leaf node. Fig. 4 shows that pair (3, 5) and (5,3) are (odd, odd) leaf nodes and (6, 2) is a leaf node although the first co-ordinate is even because as third bit and second bit of 6 (110) and 2 (10) from LSB position are 1 and 1 respectively. So the pair (6, 2) is a contradictory leaf node.

\begin{figure}[h]
\begin{center}
\includegraphics[scale=.12]{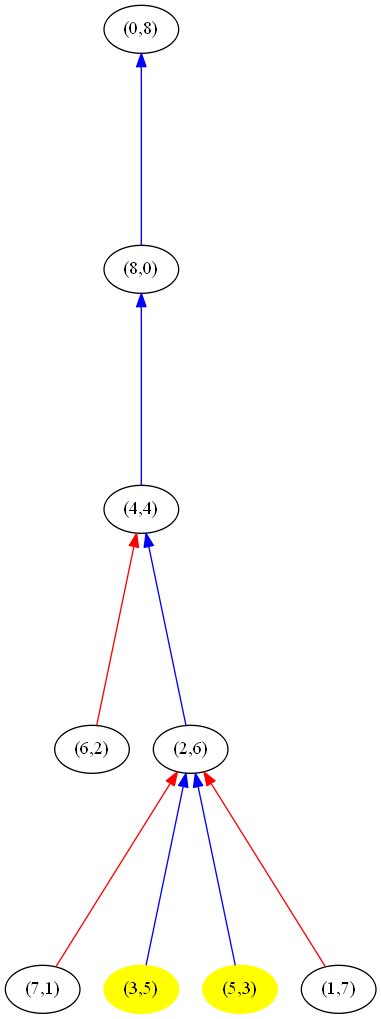}
\end{center}
\caption{CVT-XOR for an even integer 8}
\label{fig:8}
\end{figure}

\begin{theorem}
If the pair (P, Q) where the CVT value $P=a_{n+1} a_{n} a_{n-1}...a_{1}$  and the XOR value $Q=b_{n+1} b_{n} b_{n-1}...b_{1}$  is a non-leaf node in the CVT-XOR Tree and if $a_{i+1}=0$ and $b_{i}=1$ for m such I positions where $1\leq i\leq n+1$ and $m\leq n$; then the number of successors of the (P, Q) node=$2^{m}$. These nodes can be generated by substituting m column positions either ${0 \choose 1}$ or ${1 \choose 0}$ and keeping other (n-m) column values fixed.
\end{theorem}
\begin{proof}
As m column values has to be changed by ${0 \choose 1}$ or ${1 \choose 0}$ in 2 different ways, so using fundamental principle of counting the total successor pairs will be $2^{m}$.
\end{proof}
\textbf{\textit{Illustration 6:}}
Let an example of (even, even) non-leaf node (2, 6) shown below: \\
$ \begin{array}{ccccc}
& 4 & 3 & 2 & 1   \\
P=2=CVT & 0 & 0 & 1 & 0 \\
\hline  
& a= & * & * & 1   \\
& b= & ** & ** & 1  \\
\hline
Q=6=XOR &   & 1 & 1 & 0  \\
\end{array}
$\\
Here two position of a and b i.e. (*, **) in two columns can be filled up $2^{2}=4$ ways. So according to the theorem 4 such predecessor of (2, 6) are (1, 7)=(001,111), (7, 1)=(111, 001), (3, 5)=(011, 101) and (5, 3)=(101, 011) shown in fig. 4.

\begin{theorem}
Let (P, Q)$\rightarrow$(M, N), If (P, Q)  and (M, N)  are two odd-odd pair (i.e. leaf node) and immediate even-even pair respectively, then the second last LSB bit of M must be always 1 i.e.  $M=a_{n}a_{n-1}...a_{3}1 0$ and $M=M(n)=4n-2$ $\forall n$ where $a_{i}\in {0, 1}$ and $n\geq1$.
\end{theorem}
\begin{proof}
\textbf{\\ 1st part:} 
As binary representation of two odd numbers must leading to last LSB bit 1(for both P \& Q), therefore according to CVT definition second last LSB bit of M must be 1 as M=CVT(P, Q).\\
\textbf{2nd part:}  We will proof this part through the principle of method of induction.\\
For n=1, M=4-2=2, binary representation of 2 is 10 i.e. second last LSB bit is 1. So formula is true for n=1.\\
Let the formula be true for n i.e. $M(n)=4n-2$ we have to prove that this is true for m=n+1 i.e. $M^{\prime}(n+1)=4(n+1)-2$\\
$ \begin{array}{lcl}
M^{\prime}(n+1)&= & 4n-2+4 \\
&=& 4n+4-2 \\
&=& 4(n+1)-2\\
&=& M(n+1)
\end{array}
$
\\
So the formula is also be true for m=n+1. Therefor the formula is true for all values greater than equal 1. 
\end{proof}
Now we are ready to write the algorithm to check whether a given input node (X, Y) is a leaf node or not. Algorithm 2 is now modified using theorem 3 \& 4 by which a given node whether belong to the leaf or non-leaf can be verified easily is as follows:

\textbf{\textit{Algorithm 2:}}
Algorithm to Verify whether a given node is leaf or not

$ \begin{array}{lc}
 LEAF-OR-NON-LEAF(X,Y) \\
 \{  \\
 Let X=X_{n}x_{n-1}...X_{1} $and $  Y=Y_{n}Y_{n-1}...Y_{1}; \\$Where X+Y=N is an even number$.\\
 if(X==0) \\
 $Print \textquotedblleft (X, Y) is root node\textquotedblright $ \\
 elseif((X_{1}==1 \& Y_{1}==1) || (X_{1}==1 \& Y_{1}==1))  \\
 $Print \textquotedblleft (X, Y) is a (odd, odd) and (even, even) leaf  node $\\$respectively \textquotedblright $  \\
 else \\
 $Print \textquotedblleft (X, Y) is a non-leaf node \textquotedblright  $\\
 \}
 \end{array}
$

\begin{theorem}
Any Even integer (E) in the form of $E=2^{k}$ will be always a CVT value 2 for all possible (odd, odd) pairs (say (X, Y)) of E; where E=X+Y.
\end{theorem}
\begin{proof}
Notes the two facts as following:
\begin{enumerate}
\item For every pair (X, Y) of E such that if E is an n bits number then the maximum number of bit for X and Y will be n-1.
\item 2.	Every number in the form M= (2K -1) is always the addition result of two complementary numbers (say P and Q) where one of them will be even and another one is odd.
\end{enumerate}
We can represent E as:  $E= (2K -1) +1=M+1$\\
So to get E we have to add only 1 (or binary 1) to P (if P is even) or Q (if Q is even) such that P+Q=X+Y.\\
Now we get two odd numbers $p^{\prime}$ and $q^{\prime}$ such that E=M+1= $p^{\prime}+q^{\prime}$  having both LSB position are binary1 others bits position are complementary. Therefore according to the CVT definition there CVT value must be 2.
\end{proof}

Now question can be asked that what is the total number of leaf nodes and internal nodes in CVT-XOR Tree whose is root is (0, N)? 
For answering this and other similar questions we have defined three different matrices namely CVT-XOR Depth matrix, CVT-XOR Parent matrix, and CVT-XOR Frequency matrix to store depth information, parent information and frequency information respectively.

\begin{figure}[!t]
\begin{center}
\includegraphics[width=3in]{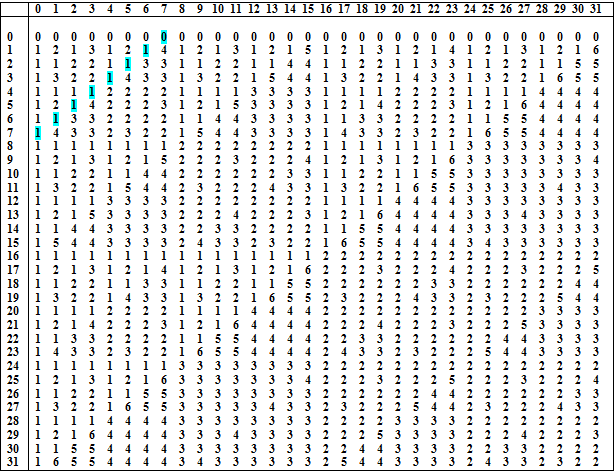}
\end{center}
\caption{CVT-XOR Depth Matrix}
\label{fig:8}
\end{figure}

\begin{figure}[!t]
\begin{center}
\includegraphics[width=3in]{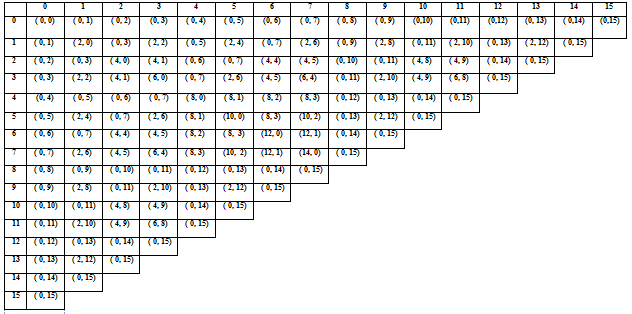}
\end{center}
\caption{CVT-XOR Parent Matrix}
\label{fig:8}
\end{figure}

\begin{figure}[!t]
\begin{center}
\includegraphics[width=3in]{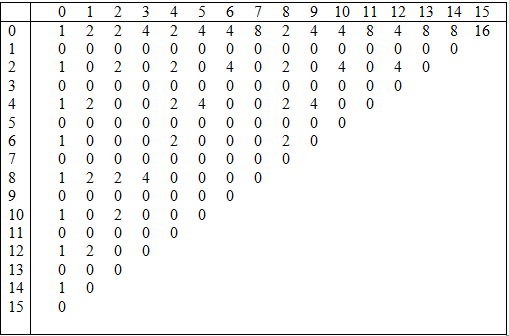}
\end{center}
\caption{CVT-XOR Frequency Matrix}
\label{fig:8}
\end{figure}

\subsection{Properties of CVT-XOR Tree obtained from the above three matrices}
\begin{enumerate}
\item The diagonal from (N, 0) to (0, N) in the CVT-XOR Depth Matrix gives the following information regarding the CVT-XOR Tree whose root is the (0, N))

\begin{enumerate}
\item	If DM (i, j)=h then (i, j) vertex is in depth-h in the CVT-XOR Tree.
\item No. of nodes in depth-h=No. of times h present on the diagonal of the CVT-XOR Depth Matrix (frequency of the occurrence of h)
\item	Average height of the Tree= (sum of all the diagonal value)/N+1
\item	If the binary pits pattern of N such that only the LSB bit is 0 i.e. number is in the form (0, 2k-2) then for such even integer N={2, 6, 14, 32,…,2k-2}), the maximum value through the diagonal from (N, 0) to (N, 0) is 2 this implies the maximum height of the Tree is always 3.
\item	If the binary pits pattern of N such all bits are 0 i.e. number is in the form (0, 2k-1) then for such odd integers N= {3, 7, 15, 31…2k-1}), the maximum value through the diagonal from (N, 0) to (N, 0) is 1 this implies the maximum height of the Tree is always 
\end{enumerate}

\item	Once we will find the number of nodes in different depth, link information ( ) among the nodes in CVT-XOR Tree can be found form the CVT-XOR Parent matrix. The diagonal from (N, 0) to (0, N) in the CVT-XOR Parent Matrix gives the following information regarding the CVT-XOR Tree whose root is the (0, N))
\begin{enumerate}
\item	If PM (i, j)=(P, Q) this implies (P, Q) is the parent of node (I, J).
\item	If (P, Q) is repeats m times in the PM matrix then there is m such similar (I, J) whose parent is (0, N)
\end{enumerate}

\item	The Frequency Matrix (FM) help us to find the following:
Let, FM (I, J) = m, where m indicate the number of predecessor or child nodes of the node (I, J).
\begin{enumerate}
\item	If I is odd then the corresponding value of m is 0, this indicate if the first co-ordinate of the node is odd and the node has no child or predecessor i.e. the node is leaf node.
\item	Through the diagonal value we can find what the leaf nodes are for a given N.
\end{enumerate}
\end{enumerate}

\subsection{CVT Operation of (odd, odd) Pairs}
Below in Fig. 8 shown the CVT value for all (odd, odd) pairs. According to the Theorem 5, CVT value for all (odd, odd) pairs in this figure contains all the numbers from the set S. Where 2 is the smallest number and the difference between two consecutive number is 4 i.e. $S=\{ 2,4, 6, 10, ...4n-2\}$. It is also showing the beautiful fractal studying in [1].
\begin{figure}[h]
\begin{center}
\includegraphics[scale=.5]{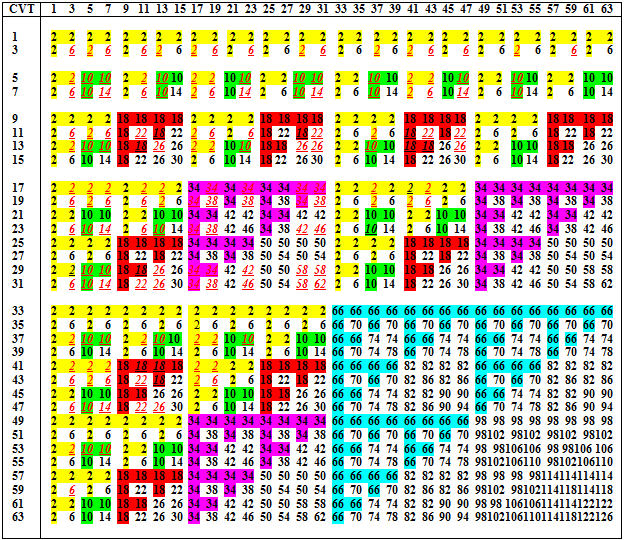}
\end{center}
\caption{CVT of all (odd, odd) pairs upto number 63}
\label{fig:8}
\end{figure}

If we traverse through a diagonal path from (K, 10) (where k is any odd number and gradually decreasing k), (k-2, 3),..., (1, k) figure (Fig. 9) like palindrome triangle is obtained. This figure demonstrate that for  any arbitrary given positive integer all the possible CVT values where bold CVT values refer to the prime pair solution.
\begin{figure}[h]
\begin{center}
\includegraphics[scale=.35]{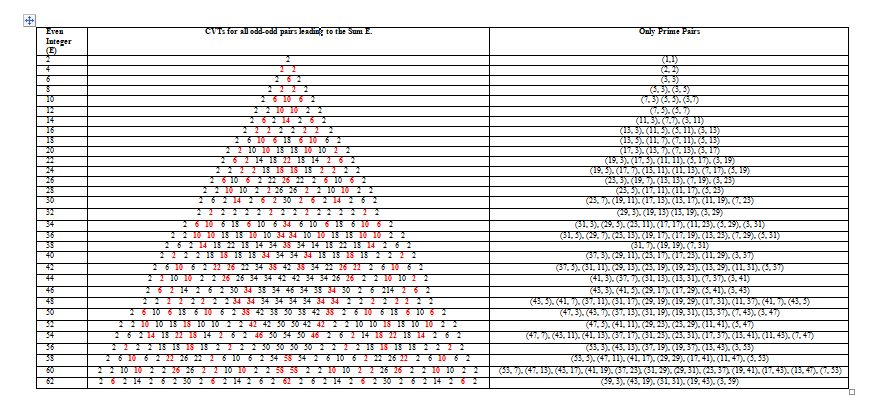}
\end{center}
\caption{Palindrome structure of immediate CVT value upto 63}
\label{fig:8}
\end{figure}

\section{Conclusion}
This paper shows a new Tree structure called CVT-XOR Tree along with two different approaches for the construction of the Tree. We also have seen different properties of the Tree like number of nodes in any depth, characterization of child, parent, internal and leaf nodes. Some other significant properties also can be obtained easily from the three different matrices along with an easy construction process of CVT-XOR Tree from the matrices. In future this work will help us in further studies in the domain of CVT-XOR to characterize the pattern formation by different CVT-XOR Trees where it is conjectured that there exists always a path from root(zero, even integer) to leaf (prime pair) signifying the Goldach Conjecture-the unfinished work remaining, how to distinguish this kind of paths from the rest. We can further explore the Tree structure by considering three numbers combination such that X+Y+Z=N and to study the different properties of such Trees.


%

\section*{Acknowledgment}
The authors would like to thank...

\ifCLASSOPTIONcaptionsoff
  \newpage
\fi

\end{document}